\documentclass[conference]{IEEEtran}
\usepackage{amsmath,amssymb,amsthm,epsfig}
\usepackage{xcolor,graphicx}

\makeatletter
\def\ps@headings{%
\def\@oddhead{\mbox{}\scriptsize\rightmark \hfil \thepage}%
\def\@evenhead{\scriptsize\thepage \hfil \leftmark\mbox{}}%
\def\@oddfoot{}%
\def\@evenfoot{}}
\makeatother
\newcommand{\Prob}[1]{\text{Pr}\{#1\}}
\newcommand{\Exp}[1]{\text{E}\{#1\}}
\newcommand{\Qcal}[1]{\mathcal{Q}(#1)}
\newcommand{\Qcalb}[1]{\mathcal{Q}\Big(#1\Big)}

\newcommand{\vb}{\mathbf{v}}

\newcommand{\qb}{\mathbf{q}}

\newcommand{\DDD}{\textsc{D$^3$sync~}}
\theoremstyle{definition}

\theoremstyle{definition}
\newtheorem{lemma}{Lemma}
\newtheorem{corollary}{Corollary}
\newtheorem{definition}{Definition}
\newtheorem{theorem}{Theorem}

\newcounter{MYtempeqncnt}

\IEEEoverridecommandlockouts
\begin{document}
\title{Discrete Dithered Desynchronization}
\author{\IEEEauthorblockN{Saman Ashkiani\\
School of Electrical and Computer Engineering \\ 
University of California, Davis \\ Email: sashkiani@ucdavis.edu}
\and
\IEEEauthorblockN{Anna Scaglione\\School of Electrical and Computer Engineering \\ 
University of California, Davis \\ Email: ascaglione@ucdavis.edu}
\thanks{This work was supported by National Science Foundation under
grant NSF-E70-8502.}
}
\maketitle
%
%
%
%

\begin{abstract}
This paper introduces the Discrete Dithered Desynchronization (\DDD) algorithm which is a
decentralized Time Division Multiple Access (TDMA) technique in which a set of network 
nodes computes iteratively a conflict-free
schedule so that each node obtains a portion of a frame that is an integer multiple of a fixed
slot size. 
The algorithm is inspired by the dynamics of Pulse Coupled Oscillators (PCO), but unlike its
predecessors that divide arbitrarily the frame among the nodes in the network, the 
\DDD allocates discrete resources among the network nodes.

Our paper proves the convergence of the \DDD algorithm and gives an upperbound on the convergence time
of the algorithm.  
\end{abstract}
\begin{IEEEkeywords}
Pulse coupled oscillators, desynchronization, and decentralized scheduling.
\end{IEEEkeywords}
%
\section{Introduction}
Agreeing on a common timing (syncrhonization) and agreeing on a Time Division Multiple Access (TDMA) schedule are related problems. In centralized scenarios there is a fixed infrastructure of nodes (masters) that optimizes the schedule and communicates them back to the other nodes (slaves). In infrastructure-less networks TDMA schedules are typically attained using decentralized graph coloring algorithms  (e.g. \cite{herman2004distributed} and \cite{rhee2006drand}). In the first algorithm time synchronization between the nodes is presumed, and the second algorithm uses extra control messages to indicate the system state.  

Synchronization schemes themselves can be divided into these two categories. Centralized algorithms rely on the broadcast signal coming from some fixed units, equipped with accurate timing devices; all nodes synchronize their clocks independently. The main example is provided by the Global Positioning System (GPS). 
In decentralized algorithms, instead, reference signals are generated at random or by a GPS signals, but then all users share information about the reference signals they sense, either by flooding the time stamp (e.g. \cite{MaKuSiLe04}) or by reaching consensus (e.g. \cite{ElGiEs02}). 

Due to the issue of coverage, centralized algorithms have a considerable infrastructure cost.
Decentralized methods are cheaper and less vulnerable to failures; however, they are slower and generally less accurate than centralized methods.

In the context of infrastructure-less networks, great attention has been paid over the past ten years to biologically inspired synchronization algorithms based on the Pulse Coupled Oscillators models \cite{strogatz2003sync}. These designs can overcome the
necessity of a reference signal for synchronization, since a reference signal emerges from the collective transmissions of the nodes. 
All PCO based algorithms, in fact, borrow a key implicit technique to communicate and update the scheduling decisions, which is through the emission of a special signal, called {\it firing} signal.
The nodes in range of the firing node, update a local state variable as they detect a firing event. Another variant of time synchronization is desynchronization \cite{nagpal}, which can be used to attain a TDMA schedule in continuous time in a decentralized fashion.  More specifically in \cite{nagpal} the authors proposed the algorithm called \textsc{Desync}  for a fully connected networks to attain a uniform TDMA schedule. Several different methodologies were utilized afterward to achieve deysnchronization, such as \cite{motskin2009lightweight} and \cite{Choochaisri2012}. In the former, desynchronization is aimed via a graph coloring algorithm using a message-passing method with a constraint on the complexity of the messages. While in the latter, the proposed algorithm uses a PCO based method with a new kind of updates based on the firing signals received from all other nodes. Other variants of the desynchronization were proposed later with various conditions and objectives. In \cite{pfs}  a \emph{Proportional Fair Scheduling} (PFS) algorithm was proposed to achieve a time schedule in which nodes are assigned a fraction of the frame proportional to their relative demand.  In \cite{anchored}, desynchronization is gained with another constraint on having a single node with a fixed time reference in the network. In general, compared to decentralized coloring, PCO based scheduling eliminates the need of synchronization to attain TDMA scheduling. However, most of the PCO based scheduling algorithms work under the condition of full connectivity in the network. In \cite{degesys2008towards} the authors tried to extend the \textsc{Desync} to multi-hop networks. 

The objective of this paper is to deal with the problem of providing a desynchronization algorithm that assigns not continuous but discrete portions of the frame. This modification is motivated by the fact that firing signals and modulated signals that are sent by the nodes are possibly of non-negligible duration compared to the frame size and, therefore, nodes cannot really use an arbitrarily short amount of the frame duration for their transmission. 

Our paper shows that, while quantizing the results of desynchronization leads to undesired fixed points for the algorithm, the algorithm we proposed, coined \emph{Discrete Dithered Desynchronization} (\DDD), converges almost surely to the desired resource allocation. A version of this protocol
was proposed in our previous work in \cite{PagliariRamy}, but  the effect of dithered quantization was assessed only by simulations.
The analysis we carry out in this paper also provides bounds on the convergence time of the algorithm. Some of our analysis takes inspiration from the methods used in \cite{quant} to prove the convergence of \emph{Quantized Consensus}. 

This paper is organized as follows, in section \ref{sec:DDD}, we introduce the \DDD algorithm. In order to analyze the characteristics of this algorithm, an equivalent consensus problem is defined in section \ref{sec:equi_prob}, and it is proved that this problem converges to a fixed group of states. In section \ref{sec:time}, an upper bound for the worst case expected convergence time of the \DDD algorithm is computed. Simulation results are shown in section \ref{sec:sim}, and a conclusion is made in section \ref{sec:conclusion}.
\section{Discrete Dithered Desynchronization Algorithm}\label{sec:DDD}
Consider a fully connected network $\mathcal{N}$ with $N$ nodes each with a discrete time counter $\Psi_i$ for $i = 1,\dots,N$. Let the time unit for counters be $\Delta$ and assume that the node counters can only take values from the finite set $\mathcal{T}_L = \{\tau_0, \dots, \tau_{L-1}\}$, where $\tau_j = j\Delta$.  Like an alarm that is advancing towards its termination, the counters advance in time as $\Psi_i(t) = \Delta(\lfloor t/\Delta\rfloor + \psi_i \pmod L)$, where $\psi_i \in \{0,\dots,L-1\}$ for all $i = 1,2,\dots,N$ denotes their initial phases. 
Each node broadcasts a firing signal to the network upon overflow in its counter, i.e. when $\Psi_i(t)$ is reset to 0 after $L\Delta$ periods. In the rest of the paper, similar to \cite{nagpal}, the time that elapses between the firing of two consecutive counters is the time allocated to each node, in the order of firing. Without loss of generality we can relabel nodes to have $0 \leq \psi_1 < \psi_2 < \dots < \psi_{N} < L$. It is implicitly assumed that no two counters can acquire the same value initially\footnote{The assumption is that in the initial random configuration no two nodes have chosen the same counter. If this condition is violated there is a collision and the nodes will try again.}. Nodes update their own counters when they sense firing signals from others and, in the absence of processing errors, desynchronization algorithms preserve the ordering in their evolution. Throughout the article we assume that all the additions and subtractions of the nodes indices are done modulo $N$, unless otherwise is stated. 

Because their initial position is random the nodes do not share a fair schedule, and the objective of \DDD update rule is to have the counters converge to a configuration for their counters within the frame approaching a uniform discrete TDMA schedule. It is conventional to describe the configuration of the counters over time around a circle that represents the frame duration, as shown in Fig. \ref{fig:network_model}.
\begin{figure}[htb]
\centering
\vspace{-0.1in}
\begin{minipage}[b]{0.8\linewidth}
\centerline{\includegraphics[width=1\linewidth]{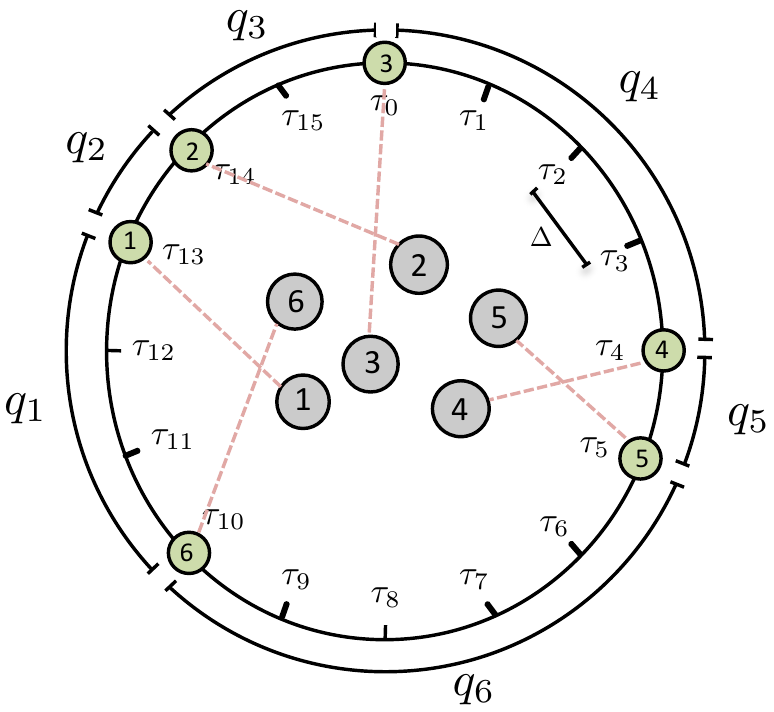}}
\hfill
\end{minipage}
\vspace{-0.2in}
\caption{Network $\mathcal{N}$ with 6 nodes at $t = t^\prime$, where $\Psi_1(t^\prime) = 13\Delta$, $\Psi_2(t^\prime) = 14\Delta$, $\Psi_3(t^\prime) = 16\Delta$, $\Psi_4(t^\prime) = 4\Delta$, $\Psi_5(t^\prime) = 5\Delta$, $\Psi_6(t^\prime) = 10\Delta$ and $L = 16$. The grey nodes inside represents a random deployment of the nodes in $\mathcal{N}$, and the green nodes on the circle are their counters around a dial, which advance in time towards the finish line, at the north pole.}
\label{fig:network_model}
\end{figure}
Given that the order is based on the temporal proximity of the counters, we will refer to as {\it time-neighbors} nodes whose indexes differ by one unit.   
However, as can be noted from  Fig. \ref{fig:network_model} the nodes geographical position is unrelated with the nodes relative counters position on the dial, which we assume to be completely random.  For simplicity, we assume that all nodes can hear each other; strictly speaking, as clarified next,
all we need is that nodes can hear their  time-neighbors.  

The \DDD~is a modification of the \textsc{Desync} algorithm. In \textsc{Desync} the update for node $i$ occurs when node $i-1
$ 
sends its firing signal at $t = t_{i-1}$, so that its counter is closer to the midpoint (in time) of the $i-1$ and $i+1$ counters:
\begin{equation}\label{eq.one}
\widehat{\Psi}_{i}(t_{i-1}^+) = \alpha \Psi_{i}(t_{i-1}) + \frac{1-\alpha}{2}\Big(\Psi_{i+1}(t_{i-1}^+)+\Psi_{i-1}(t_{i-1}^+)\Big),
\end{equation}   
where $\alpha \in (0,1)$, and $t_{i-1}^+$ denotes the time right after the updates resulted from node $i-1$ firing signal. Because node $i-1$ just fired, its counter is reset at time $t_{i-1}^+$, i.e. $\Psi_{i-1}(t_{i-1}^+) = 0$. Also, since only node $i$ updates its counter in the network at $t_{i-1}$, the counter of node $i+1$ remains unchanged $\Psi_{i+1}(t_{i-1}^+) = \Psi_{i+1}(t_{i-1})$. Therefore, \eqref{eq.one} is equivalent to:
\begin{equation}\label{eq:desync_psi1}
\widehat{\Psi}_{i}(t_{i-1}^+) = \alpha \Psi_{i}(t_{i-1}) + \frac{1-\alpha}{2}\Psi_{i+1}(t_{i-1}).
\end{equation}
A simple way to modify this algorithm to obtain a discrete schedule so that $\widehat{\Psi}_{i}(t_{i-1}^+)\in \mathcal{T}_L$ is to uniformly quantize the counter with the mapping $Q(x) = \min_{\tau_j \in \mathcal{T}_L} {|\tau_j - x|}$. However, as it will be discussed in Lemma \ref{lemma:uniform_q}, this approach may not  converge to a desired fixed point. Hence, in \DDD we used a randomized $\mathcal{Q}$, which is a dithered quantization \cite{wannamaker2000theory} over $\mathcal{T}_L$ defined as follows:
\begin{equation}\label{eq:dither}
\Qcal{x} \triangleq \min_{\tau_j \in \mathcal{T}_L} {|\tau_j - (x+\vb)|},
\end{equation} 
where $x \in \mathbb{R}$ and $\vb$ be a random variable, uniformly distributed over $(-\Delta/2,\Delta/2)$. $\vb$ and $x$ are statistically independent. Let $\tau_j \leq x< \tau_{j+1}$, then from equation (\ref{eq:dither}) distribution of  $\Qcal{x}$ can be written as $\Prob{\Qcal{x} = \tau_i}$
\begin{equation}\label{eq:DQ}
=
\begin{cases}
\Prob{x+\vb<\tau_i+\Delta/2} = 1- \frac{x-\tau_i}{\Delta} & \text{if }i = j,\\
\Prob{x+\vb>\tau_{i-1}+\Delta/2} = \frac{x-\tau_{i-1}}{\Delta} & \text{if } i = j+1,\\
0 & \text{otherwise}.
\end{cases}
\end{equation}
The idea is analogous to the \emph{probabilistic quantization} used in \cite{can} to ensure the convergence of an average quantized consensus policy.

Using dithered quantization (\ref{eq:dither}) on (\ref{eq:desync_psi1}) $\widehat{\Psi}_{i}(t_{i-1}^+)$ is mapped onto $\mathcal{T}_L$:
\begin{equation}\label{eq:desync_psi2}
\Psi_i(t_{i-1}^+) = \Qcal{\widehat{\Psi}_{i}(t_{i-1}^+)}.
\end{equation}
Thus, \emph{Discrete Dithered Desynchronization} (\DDD)  algorithm works as shown in Fig. \ref{fig:DDDesync1} and {\ref{fig:DD_update}}.
\begin{figure}[htb]
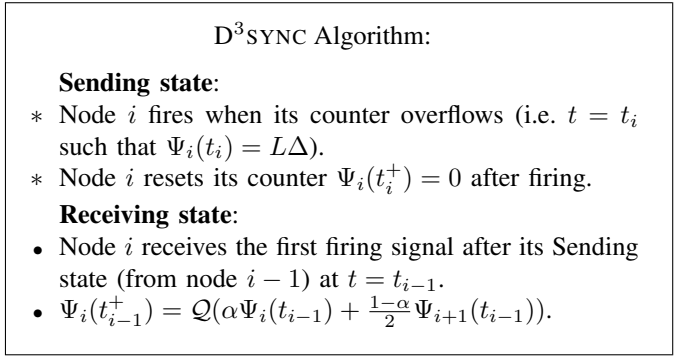

\centering \vrule
\begin{minipage}{\linewidth}
\hrule \vspace{0.25cm}
\begin{minipage}{0.95\linewidth}
\def\baselinestretch{1} 
\begin{center} \DDD Algorithm: \end{center}
\begin{itemize}
\item[] {\bf Sending state}:
\item[$\ast$] Node $i$ fires when its counter overflows (i.e. $t = t_{i}$ such that $\Psi_{i}(t_{i}) = L\Delta$).
\item[$\ast$] Node {$i$} resets its counter $\Psi_{i}(t_{i}^+) =0$ after firing.
\end{itemize}
	\begin{itemize} 
	\item[] {\bf Receiving state}:
		\item  Node $i$ receives the first firing signal after its Sending state (from node $i-1$) at \mbox{$t=t_{i-1}$}.
		\item $\Psi_i(t_{i-1}^+) = \Qcal{\alpha \Psi_{i}(t_{i-1}) + \frac{1-\alpha}{2}\Psi_{i+1}(t_{i-1})}$.
	\end{itemize}
\end{minipage}
\vspace{1em} \hrule
\end{minipage}\vrule \\
\caption{The \DDD algorithm for node $i$. When its counter overflows (Sending state), or when it receives the firing signal (Receiving state). Note that all the node needs is to follow the evolution of its own counter and wait until it hears a node fire after its own firing event. The use
of indexes to refer to the nodes is for the convenience of our description, but it is totally irrelevant for the update.}\label{fig:DDDesync1}
\end{figure}
\begin{figure}
\centering
\vspace{-0.15in}
\begin{minipage}[b]{0.6\linewidth}
\centerline{\includegraphics[width=1\linewidth]{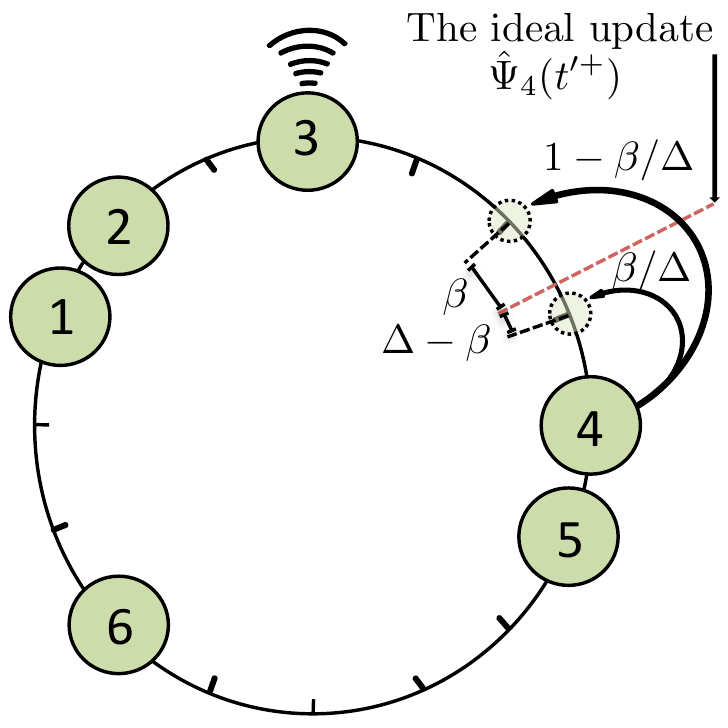}}
\hfill
\end{minipage}
\vspace{-0.3in}
\caption{\DDD updates: Time evolution model of network $\mathcal{N}$ with 6 nodes at $t = t^\prime$, where $\Psi_1(t^\prime) = 13\Delta$, $\Psi_2(t^\prime) = 14\Delta$, $\Psi_3(t^\prime) = 16\Delta$, $\Psi_4(t^\prime) = 4\Delta$, $\Psi_5(t^\prime) = 5\Delta$, $\Psi_6(t^\prime) = 10\Delta$ and $L = 16$.} 
\label{fig:DD_update}
\end{figure}

It should be clear at this point that, if there are no errors, the order of firings will be preserved. In fact, if initially all counters have at least a unit difference between each other, by using the semi linear update (\ref{eq:desync_psi1}) the updated counter cannot cross the two fixed time-neighbors which are both on $\mathcal{T}_L$. In other words, the midpoint never falls inside an adjacent quantization bin of any two time-neighbors, and hence it is impossible to be quantized to any of their values.

It is important to note that node $i$ does not have direct access to other nodes' counters (i.e. $\Psi_{i-1}$ and $\Psi_{i+1}$). The only information it acquires is the counter difference between its own counter and the node which last fired, prior to its own firing event, and hence immediately after it. In particular in the \DDD algorithm, node $i$ can compute the exact counter difference with node $i-1$ and only an estimation of the counter difference with node $i+1$. For the rest of this article, we assume that node $i$ has exact information about the counter difference with the two time-adjacent neighbors (i.e. nodes $i-1$ and $i+1$)\footnote{ The assumption amounts to considering an ideal Physical (PHY) Layer, that detects and estimates perfecting the epoch of the firing signal. Simulations done in \cite{pfs} have shown that the algorithm is robust and can recover from errors (false alarms or missed detection).}. 

The key ingredient of our incoming analysis, is a mapping of the counters used in \DDD onto a set of auxiliary variables. In fact, for the sake of the analysis it is preferable to consider the counter differences instead of the absolute counter values, since the sought TDMA schedule is equivalent to attain consensus on the counter differences. More specifically, let $\qb(t) = (q_1(t),\dots, q_{N}(t))^T$ denotes the difference between consecutive nodes counters, defined as follows:
\begin{equation}\label{eq:def_q}
q_i(t) \triangleq \Psi_i(t) - \Psi_{i-1}(t) \pmod L.
\end{equation} 
In order to rewrite the update equations based on the counter differences, after substituting (\ref{eq:def_q}) in (\ref{eq:desync_psi1}) and (\ref{eq:desync_psi2}):
\begin{equation}\label{eq:desync_q2}
q_i(t_{i-1}^+) = \Qcalb{\frac{1+\alpha}{2}q_{i}(t_{i-1})+ \frac{1-\alpha}{2}q_{i+1}(t_{i-1})}.
\end{equation}
It is clear from the definition in (\ref{eq:def_q}) and Fig. \ref{fig:network_model} that for any $t$ the sum of all counter differences remain constant 
\begin{equation}\label{eq:q_const}
\sum_{i=1}^{N}{q_i(t)} = L\Delta.
\end{equation}
Hence, as $q_i$ and $q_{i+1}$ are the only participants in the update equation (\ref{eq:desync_q2}) at time $t = t_{i-1}$, the sum of the two counter differences should be preserved after the update:
\begin{equation}\label{eq:q_const2}
q_i(t_{i-1}^+)+q_{i+1}(t_{i-1}^+) = q_i(t_{i-1})+q_{i+1}(t_{i-1}).
\end{equation}

As it was stated before, our final objective toward the implementation of a TDMA schedule over discrete resources is to divide a finite number of time slots equally among all the nodes in the network. But, this is only possible when $L = \ell N$ for some integer $\ell$, and therefore each node will have $\ell$ slots (i.e. $q_i = \ell \Delta$ for all $i = 1,\dots,N$). We define more precisely what we mean by a TDMA schedule in general in the following definition:
\begin{definition}
\emph{TDM state}: If $L = \ell N+r$ for $0 \leq r < N$, then we denote TDM states as all subdivisions of the slots among the nodes in which $r$ of nodes have $\ell + 1$ slots and others have $\ell$ slots.  
\end{definition}
In the following section we define an equivalent description of the evolution of the network state that eases the analysis of the behavior of the \DDD algorithm. The goal of our analysis is first to prove that the network dynamics will converge to one TDM state from any initial condition. We also provide a bound on the longest expected amount of time it takes to reach a TDM state, starting from the set of worst initial conditions.  
\section{Analysis of the equivalent consensus problem}\label{sec:equi_prob}
As stated at the end of the previous section, in \DDD the values of the counters' differences in (\ref{eq:desync_q2}) at any firing event fully determine the evolution of the network $\mathcal{N}$ schedule over time. It is assumed that all counters have a synchronized rise edge. Thus, $q_i$ (which is the counter difference between the pair of nodes $i-1$ and $i$) may only change right after the firing events at $t_{i-1}$ or $t_{i-2}$, and it remains unchanged for all other firing events until the next round. Hence, we state that \DDD corresponds to the equivalent consensus problem described next. 
Suppose $\mathcal{G}$ is a ring graph with $N$ vertices $\mathcal{V} = \{v_1,v_2,\dots,v_N\}$, each possessing a value $q_i$ for $i = 1,2,\dots,N$. These vertices are connected to each other by the ring graph $\mathcal{G}$ as depicted in Fig. \ref{fig:DD_eq}, where a single edge, named as \emph{active edge}, is distinguished from all other \emph{inactive} edges. The two vertices connected to the active edge are called active nodes. The active edge also rotates counter clock-wise. Let each disposition of the active edge in time represents an event in $\mathcal{G}$ denoted as $\{n_1,n_2,n_3, \dots\}$. For instance, if $(v_{i},v_{i+1})$ is an active edge at $n_j$, then $(v_{i-1},v_{i})$ will be an active edge at $n_{j+1}$. Now, if $q_i \in \mathcal{T}_L$ for all $i$ and $\sum_{i=1}^{N}{q_i} = L\Delta$, and also at each event active nodes perform an interaction on their values based on update equations similar to (\ref{eq:desync_q2}) and (\ref{eq:q_const2}), then this consensus problem would be equivalent to performing \DDD algorithm on network $\mathcal{N}$. In general if $(v_{i},v_{i+1})$ is an active edge at $n_j$, then the stored values $q_i$ and $q_{i+1}$ are updated as follows:
\begin{equation} \label{eq:active1}
q_i^\prime = \Qcalb{\frac{1+\alpha}{2}q_{i}+ \frac{1-\alpha}{2}q_{i+1}},
\end{equation}
\begin{equation}\label{eq:active2}
q_{i+1}^\prime = q_i+q_{i+1} - q_i^\prime,
\end{equation}
where $\alpha \in (0,1)$, and $q_i^\prime$ and $q_{i+1}^\prime$ represent the updated values of $q_i$ and $q_{i+1}$ respectively and right before $n_{j+1}$. Note that for any $x \in \mathcal{T}_L$ and $y \in \mathbb{R}$, it follows that $\Qcal{x+y} = x + \Qcal{y}$. Thus, as $q_{i+1} \in \mathcal{T}_L$, (\ref{eq:active1}) can be written as
\begin{align}\label{eq:active1_2}
q_i^\prime & = \Qcalb{\frac{1+\alpha}{2}q_i + (1 -\frac{1+\alpha}{2})q_{i+1}} \nonumber \\
& = q_{i+1} + \Qcalb{\frac{1+\alpha}{2}(q_i - q_{i+1})}.
\end{align}
\begin{figure}
\centering
\begin{minipage}[b]{0.6\linewidth}
\centerline{\includegraphics[width=1\linewidth]{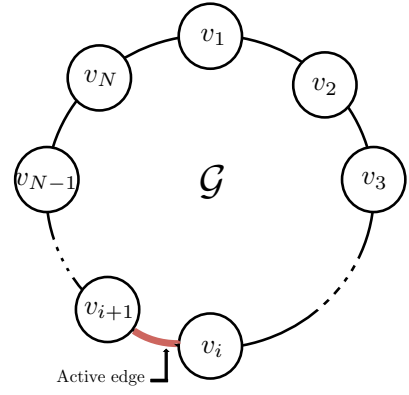}}
\hfill
\end{minipage}
\vspace{-0.2in}
\caption{The equivalent consensus problem for graph $\mathcal{G}$ with N nodes. The active edge is currently $(v_i,v_{i+1})$, and it will rotate in a counter clock wise direction at each time step.} 
\label{fig:DD_eq}
\end{figure}

In the following parts, we investigate the behavior of the network's state $\qb = (q_1,\dots,q_N)^T$ in an evolving graph $\mathcal{G}$. Our objective is to show that the \DDD algorithms converges to the TDM states almost surely. The effect of uniform quantization in \eqref{eq:active1} is discussed in the following lemma.
\begin{lemma}\label{lemma:uniform_q}
In the \DDD equivalent problem, if uniform quantization is used instead of the dithered quantization. Then there exists cases where the algorithm converges to a non-TDM state.
\end{lemma}
\begin{proof}
The proof for this lemma is in Appendix \ref{App:uniform}.
\end{proof}
\begin{definition}
Let $(v_{i},v_{i+1})$ be an active edge in graph $\mathcal{G}$ at $n_j$. Then, the possible interactions are:
\begin{enumerate}
\item \emph{Null}: When $q^\prime_i = q_i$ and $q^\prime_{i+1} = q_{i+1}$.
\item \emph{Swap}: When $q^\prime_i = q_{i+1}$ and $q^\prime_{i+1} = q_{i}$.
\item \emph{Compression}: When $|q^\prime_i - q^\prime_{i+1}| < |q_{i} - q_{i+1}|$.   
\end{enumerate}
\end{definition}
The following lemma discusses the behavior of the interactions in the \DDD equivalent consensus problem.
%
%
\begin{lemma}\label{lemma:DD_prop}
Let $(v_{i},v_{i+1})$ be an active edge in graph $\mathcal{G}$ and $q_i,q_{i+1}\in \mathcal{T}_L$, then:
\begin{enumerate}
\item If $q_i = q_{i+1}$, with probability one ${q}_i^\prime = {q}_{i+1}^\prime$.
\item  If $|q_i - q_{i+1}| = \Delta$, the updated values will have the same difference afterward $|q_i^\prime - q_{i+1}^\prime| = \Delta$.
\item If $|q_i - q_{i+1}| > \Delta$: 
	\begin{enumerate}
	\item If $q_i> q_{i+1}$: 
	\begin{equation}
	q_i^\prime \leq q_i \quad \text{and} \quad q_{i+1}^\prime \geq q_{i+1}.
	\end{equation}
	\item If $q_i < q_{i+1}$:
	\begin{equation}
	q_i^\prime \geq q_i \quad \text{and} \quad q_{i+1}^\prime \leq q_{i+1}.
	\end{equation}
	\end{enumerate}
\item After each interaction $|q_i^\prime - q_{i+1}^\prime| \leq |q_i - q_{i+1}|$.
\end{enumerate}
\end{lemma}
\begin{proof} The proof for this lemma is in Appendix \ref{app_1}. \end{proof}
This lemma characterizes all the possible interactions in $\mathcal{G}$. Let $(v_{i},v_{i+1})$ be an active edge in $\mathcal{G}$. If both of the active vertices have the same values, with probability one the null interaction will happen. Also, if the values of the active vertices have more than one unit difference (like in Lemma \ref{lemma:DD_prop} case 3), the smaller one will increase and the bigger one will decrease, i.e. the interaction is a compression.
It is important to note that in general the difference between the values would decrease with a positive probability, which depends on the choice of $\alpha$.
\begin{corollary}\label{cor.prob}
If active nodes have a difference of $\Delta$ in their values (Lemma \ref{lemma:DD_prop} case 2) two cases may happen: 1) a null interaction, with probability $\frac{1+\alpha}{2}$, or 2) a swap interaction, with probability $\frac{1-\alpha}{2}$. 
\end{corollary}

Based on the results of Lemma \ref{lemma:DD_prop}, it is simple to see that the properties of the \DDD equivalent consensus problem are similar to the Quantized Consensus in \cite{quant} on the ring graph $\mathcal{G}$. Reaching an average quantized consensus over $\mathcal{G}$ is equivalent to having the desired discrete desynchronization in the original network $\mathcal{N}$. There are two important differences, however, between the \DDD algorithm and the Quantized Consensus: 1) The edge selection in our algorithm is deterministic (i.e. the active edge deterministically rotate counter clockwise), while in the Quantized Consensus the edges are chosen at random. 2) In our algorithm the updates are randomized, similar to \cite{can}, while in Quantized Consensus there is a deterministic operation over the chosen edge. The random evolution of the network state $\qb = (q_1,\dots,q_N)^T$ in $\mathcal{G}$ is a Markov chain. In fact, to specify the distribution of the next state $\qb^\prime$ all is needed is the current position of the active edge and the current network state $\qb$ (in particular the values of the active nodes). Also in the equivalent consensus problem $\sum_{i = 1}^{N}{q_i} = L\Delta$ and $q_i \geq \Delta$ (based on the constraints of the \DDD algorithm). Consequently, there are finite number of possible states for network $\mathcal{G}$. Because of the deterministic movement of the active edge, \DDD equivalent consensus problem  is a cyclo-stationary Markov process. In contrast, the Quantized Consensus in \cite{quant} and the average consensus method in \cite{can} form stationary Markov processes, and hence the evolution can be modeled as a homogenous Markov chain. This difference is particularly important in characterizing the convergence time, which is more complex for \DDD than for Quantized Consensus. Establishing almost sure convergence is simpler, and done to prove the following results.



\begin{theorem}\label{thm:AS}
In the \DDD equivalent consensus problem on network $\mathcal{G}$, TDM states are absorbing states.
\end{theorem}
\begin{proof}
Let $\mathcal{S}_L \triangleq \{\qb: |q_i - q_j| \leq \Delta, \forall v_i,v_j \in \mathcal{G}, \|\qb\|_1 = L\Delta\}$ denotes the set of all TDM states in $\mathcal{G}$.
It suffices to show that once the network is in a TDM state (i.e. $\qb \in \mathcal{S}_L$), the updated network will also be in one of the TDM state (i.e. $\qb^\prime \in \mathcal{S}_L$).
Suppose $\qb \in \mathcal{S}_L$ and $(v_i,v_{i+1})$ is the active edge. Then for those interactions in which both vertices of the active edge have equal value (i.e. $q_i = q_{i+1}$), based on Lemma~\ref{lemma:DD_prop} case 1, $q_i^\prime = q_{i+1}^\prime$. Therefore, $\qb^\prime = \qb \in \mathcal{S}_L$. However, for those interactions in which $|q_i - q_{i+1}| = \Delta$, based on Lemma~\ref{lemma:DD_prop} case 2, $|q_i^\prime - q_{i+1}^\prime| = \Delta$, which means either a null or swap interaction will occur for $q_i$ and $q_{i-1}$. In either case, all the other inactive vertices will have their values unchanged and, hence, $\qb^\prime \in \mathcal{S}_L$, which proves the statement. 
\end{proof}

\begin{theorem}\label{thm:convergence}
In \DDD equivalent consensus problem on network $\mathcal{G}$, the TDM states are the only absorbing states and the network converges to the TDM states almost surely.
\end{theorem}
\begin{proof}
Let $\qb(n_0)$ be the initial state of the network $\mathcal{G}$. We want to show that after sufficient number of interactions, with probability one, the network reaches the absorbing state:
\begin{equation}
\lim_{n \to \infty}{\Prob{\qb(n) \in \mathcal{S}_L}} = 1.
\end{equation} 
Let $R_\mathcal{G}(n) \triangleq \max_i{q_i(n)} - \min_j{q_j(n)}$ denote the range between the maximum and minimum $q_i(n)$ in network $\mathcal{G}$ at the $n$th time step. 
Based on Lemma~\ref{lemma:DD_prop} case 4, after each interaction the difference between the values of the two active vertices is a non-increasing quantity, i.e. $R_\mathcal{G}(n+1) \leq R_\mathcal{G}(n)$. In particular, $R_\mathcal{G}(n+1) < R_\mathcal{G}(n)$ when there is a single node with the maximum (or minimum) value in the network which participate in a compression. Also, $0 \leq R_\mathcal{N}(n) \leq \Delta$ only when the network is in $\mathcal{S}_L$; in fact, $R_\mathcal{G}(n) = 0$ only when $L = \ell N$ and $R_\mathcal{G}(n) = \Delta$ otherwise. Instead, $R_\mathcal{N}(n) > \Delta$ for all the states which are not in $\mathcal{S}_L$. Therefore, to prove the theorem, it is sufficient to show that while $\qb(n) \notin \mathcal{S}_L$, $\Prob{R_\mathcal{G}(n+k) < R_\mathcal{G}(n)} \geq 0$, for some $k$. In other words, it is possible for the network to move into a state with lower range after finite number of interactions. Suppose the network state is not in $\mathcal{S}_L$. As the active edge rotates over the entire network from this initial state, there exists a time step $n^\ast$ in which one of the two vertices of the active edge has the value $q_\text{max} = \max_i{q_i(n^\ast)}$. We denote that active vertex as $v_M$. Since, $\qb(n^\ast) \notin \mathcal{S}_L$, there should exist at least one other node in $\mathcal{G}$ with value lower than $q_\text{max}-\Delta$, otherwise the initial network state considered must be in $\mathcal{S}_L$, which is a contradiction. Now moving counter clockwise from $v_M$, let $v_j$ be the closest node such that: 
\begin{equation}
j = M - \arg \min_{k}{\{q_{M-k} < q_\text{max}-\Delta\} \pmod N}.
\end{equation}
Then all the nodes between $v_M$ and $v_j$ (moving counter clockwise) have values equal to either $q_\text{max}$ or $q_\text{max}-\Delta$. Based on Lemmas \ref{lemma:DD_prop}-1 and \ref{lemma:DD_prop}-2, with positive probability the  value of $v_M$ can swap until  $q_{j+1} = q_\text{max}$, and then, with positive probability, a compression will happen on the active edge $(v_{j},v_{j+1})$. Depending on the number of nodes having the value equal to $q_\text{max}$, by following the previous steps, one can prove that, with positive probability, the number of nodes with maximum value is reduced by one, until all of them have participated in a compression, which is the case where $R_\mathcal{N}$ is strictly decreased. Thus, independent of the initial state, with positive probability $R_\mathcal{N}$ is decreased step by step and with a finite number of interactions, until it reaches the minimum value which correspond to an absorbing state.  Since, the sum of the $q_i$ is preserved during each interaction, it is not necessary to keep track of the minimum values, but the same process happens on the minimum values too.

It was shown that, starting from any state that $R_\mathcal{G} > \Delta$ (i.e. states which are not in $\mathcal{S}_L$), with positive probability $R_\mathcal{G}$ is strictly decreased after sufficient number of interactions. Thus, the only possible absorbing states are $\mathcal{S}_L$ which $0 \leq R_\mathcal{G} \leq \Delta$.
\end{proof}
\section{Analysis of the Convergence Time for the \DDD Algorithm}\label{sec:time}
Up to this point, it was shown that irrespective of the initial network state (i.e. $\qb(n_0)$), after a sufficient number of interaction the system reaches a TDM state in $\mathcal{S}_L$ almost surely. It is now of interest to evaluate the expected number of interactions  needed for a system to reach a state in $\mathcal{S}_L$. 

Let $\mathcal{S} = \{\qb: \sum_{i=1}^{N}{q_i} = L\Delta, q_i \geq \Delta, \forall v_i \in \mathcal{V}\}$ be the set of all possible states in the network $\mathcal{G}$, and \mbox{$V(\qb) = \|\qb - \bar{q}\mathbf{1}\|^2/\Delta^2$}, where $\bar{q} = \frac{1}{N}\sum_{i=1}^{N}q_i$. A similar function is defined in \cite{quant} to characterize the dynamics of the average consensus with quantized values.
Similar to Lemmas 4 and 5 in \cite{quant}, it can be shown that:
\begin{lemma}\label{remark1}
If $\qb \in \mathcal{S}$, $V(\qb)$ has the following properties for the network $\mathcal{G}$ with $\|\qb\|_1 = L\Delta$:
\begin{enumerate}
\item $\max_{\qb \in \mathcal{S}}{V(\qb)} \leq (L-N)^2N/4$.
\item After each interaction: 
\begin{equation}
\begin{cases}
V(\qb^\prime) = V(\qb) \ \text{if} \ |q_i - q_{i+1}| \leq \Delta,\\
V(\qb^\prime) \leq V(\qb) -2 \ \text{otherwise}.
\end{cases}
\end{equation}
\end{enumerate} 
\end{lemma}
\begin{proof} As $\sum_{i=1}^{N}{q_i(t)} = L\Delta$, and $q_i(t) \geq \Delta$, $\max_i{q_i(t)} \leq (L-N+1)\Delta$, which can only occur when all other values are equal to $\Delta$. Thus, the first property is achieved by modifying Lemma 4 in \cite{quant} by the following bounds. The second property is another representation of Lemma 5 in \cite{quant} for this scenario.\end{proof}

$V(\qb(n))$ is a positive function which is decreasing in $n$, and will reach its minimum once the system reaches the TDM states (i.e. $\qb(n) \in \mathcal{S}_L$). The intuition behind this fact is that once the system is in $\mathcal{S}_L$, all $q_i(n)$ are either equal or can only have a unit slot difference, then from the second property of Lemma \ref{remark1} $V(\qb)$ does not change anymore. Also, from Theorem \ref{thm:convergence} the system will eventually reach the TDM states in $\mathcal{S}_L$, equivalently $V(\qb)$ will reach its minimum.

As it was mentioned earlier, the dynamics of the network $\mathcal{G}$ can be modeled as a Markov chain. Let $T(n_0) \triangleq \min \{n \geq n_0: \qb(n_0) \notin \mathcal{S}_L, \qb(n) \in \mathcal{S}_L\}$ be the absorption time of the Markov chain. Let ${T_1}(\qb, n^\ast)$  be a random variable defined as the number of interactions, starting from step $n^\ast$ and state $\qb(n^\ast)$, until the first compression in the network $\mathcal{G}$. Let $\bar{T}(n_0) = \max_{\qb \in \mathcal{S}}{\Exp{T_1(\qb,n_0)}}$ be the maximum expected number of interactions one has to wait until network $\mathcal{G}$ experiences a compression. Then based on the results from Lemma \ref{remark1}, it can be shown just like in \cite{quant} that:
\begin{equation}\label{eq:bounds1}
 \max_{\qb(n_0) \in \mathcal{S}}{\Exp{T(n_0)}} \leq \bar{T}(n_0)(L-N)^2N/8.
\end{equation}
The intuition behind the inequality (\ref{eq:bounds1}) is as follows: Based on the second property of Lemma \ref{remark1}, depending on the initial state $\qb(n_0) \notin \mathcal{S}_L$,  for $n>n_0$,  $V(\qb(n))$ will decrease only if a successful compression happens. In the worst case, $V(\qb(n_0))$ is maximum and at each compression it is decreased by the smallest possible amount, which is by $2$. Also in this scenario, two consecutive compressions  happen after at most $\bar{T}(n)$ number of interactions (in expectation), given the fact that at the $n$th time step a compression occurred. Next we provide the main analytical result of this paper:
\begin{theorem}\label{thm:timebound}
In the network $\mathcal{G}$ with $N$ nodes and update parameter $\alpha \in (0,1)$, the worst case maximum expected number of null and swap interactions until a compression occurs, $\bar{T}(n_0)$, can be calculated for large $N$ as:
\begin{equation}\label{eq:thm3}
\bar{T}(n_0) = \mathcal{O}\Big(g(\alpha)N^3\Big),
\end{equation}
where $g(\alpha) = \frac{(\alpha+1)}{24(1-\alpha)}$. The exact value of $\bar{T}(n_0)$ for any $N$ can be found in equation (41) in Appendix \ref{app_2}.
\end{theorem}
\begin{proof} The proof of this theorem is in Appendix \ref{app_2}.\end{proof}
Based on Theorems \ref{thm:timebound} and (\ref{eq:bounds1}), the greatest expected number of interactions to reach a TDM state  is upper bounded by $\mathcal{O}\big(g(\alpha)(L-N)^2N^4 /8\big)$. Also, $g(\alpha)$ is a constant factor which is an strictly increasing function of $\alpha$ for $\alpha \in (0,1)$. As it will be shown in simulation results in Section \ref{sec:sim}, in general this upper bound is not tight for the \DDD algorithm, and it represents an upper bound for the worst case possible.
\section{Simulation Results} \label{sec:sim}
In this section performance of the \DDD algorithm in TDMA scheduling and also its convergence time is assessed by using computer simulations.

The \DDD algorithm has a single parameter $\alpha$ to be chosen. It represents the algorithm's inertia in updating the current phase counters based on the received firing signal. As it is clear from \eqref{eq:thm3}, large $\alpha$ shows more resistance toward the update, and hence its convergence time is greater in comparison with a smaller $\alpha$.

In Fig. \ref{fig:Convergence}, convergence of the \DDD algorithm to the TDM states is shown under two different scenarios. Network $\mathcal{N}$ is consisted of $6$ nodes with $\alpha = 0.2$. In (a) $L = 60$ and the network has a single TDM state (i.e. $|\mathcal{S}_L| = 1$), while on (b) $L = 57$ and hence there is not a single TDM state. After six round of firings, where each round is defined as a complete cycle of firing events in the network and by all of the nodes, the network is absorbed into $\mathcal{S}_L$. There are $|\mathcal{S}_L| = 20$ different TDM states and the network changes its state among these states afterward.

\begin{figure}
\centering
\begin{minipage}[b]{0.48\linewidth}
\centerline{\includegraphics[width=1.0\linewidth]{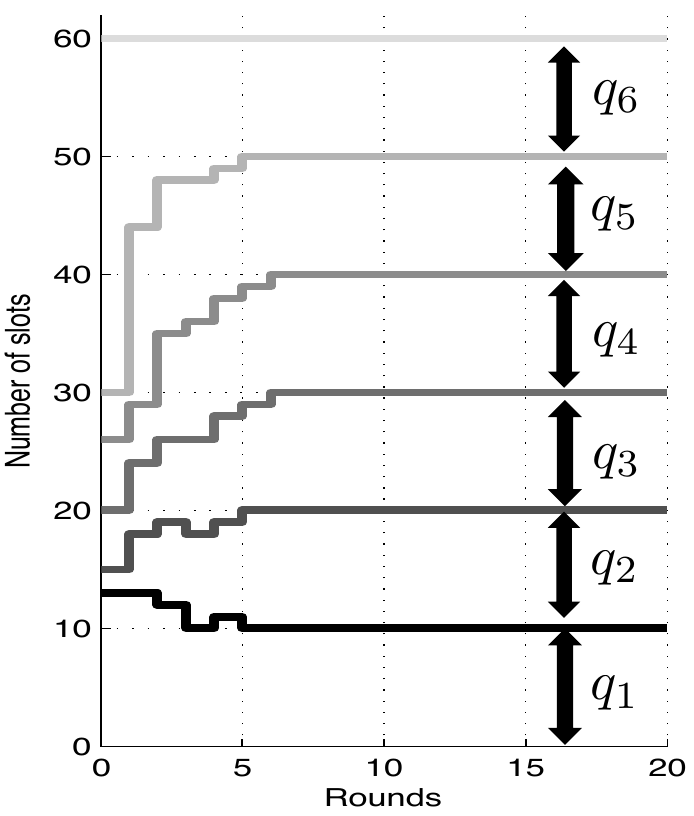}}
\centerline{(a)}\medskip
\hfill
\end{minipage}
\begin{minipage}[b]{0.48\linewidth}
\centerline{\includegraphics[width=1.0\linewidth]{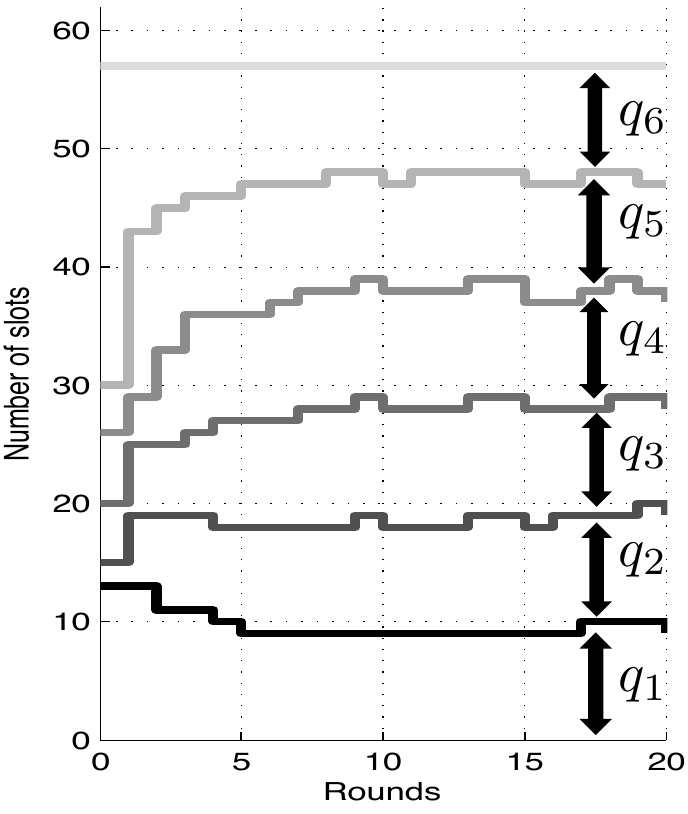}}
\centerline{(b)}\medskip
\hfill
\end{minipage}
\vspace{-0.35in}
\caption{Convergence of the \DDD algorithm over two scenarios for a network $\mathcal{N}$ with 6 nodes and $\alpha =0.2$: (a) $L = 60$, where there exists a unique absorbing state. (b) $L = 57$, where $|\mathcal{S}_L| = 20$. Both scenarios are converged after 6 rounds of firings.} 
\label{fig:Convergence}
\end{figure}
Number of interactions until absorption for the \DDD algorithm is depicted in Fig. \ref{fig:time_Convergence}-a. As it is clear from the figure, the upper bound in \eqref{eq:bounds1} is very pessimistic, because it is based on the performance of the worst possible case  as for the initial state. Theorem \ref{thm:timebound} found an expression for the expected number of interaction until the next compression in network $\mathcal{G}$. As it is explained in Appendix \ref{app_2}, this is the case where all nodes have a common value, except for the two of them, one of which a unit higher and the other a unit lower than the common value (also shown in \eqref{eq:App_dist}). The average number of interactions until absorption for this case is shown in Fig. \ref{fig:time_Convergence}-b. The evaluated absorption time in Appendix \ref{app_2} which was discussed in Theorem \ref{thm:timebound} is also shown in this figure by dashed lines. 
\begin{figure}
\centering
\begin{minipage}[b]{0.48\linewidth}
\centerline{\includegraphics[width=1.0\linewidth]{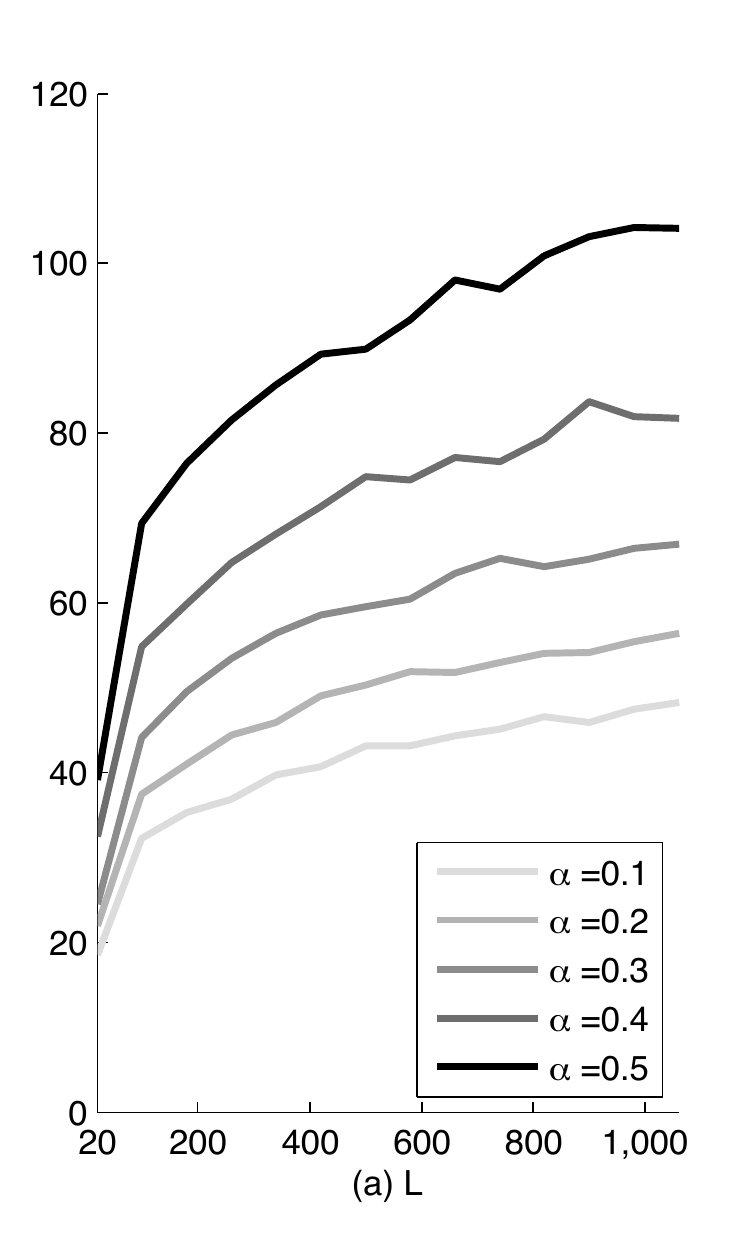}}
\hfill
\end{minipage}
\begin{minipage}[b]{0.48\linewidth}
\centerline{\includegraphics[width=1.0\linewidth]{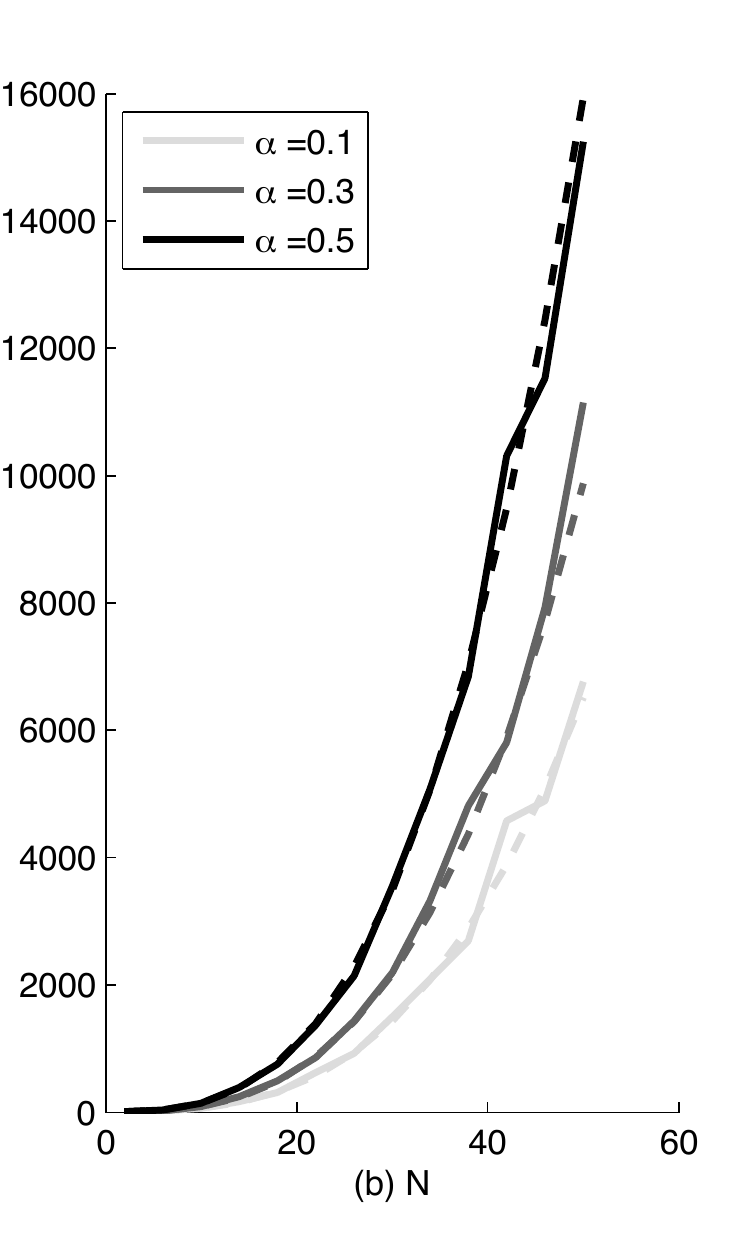}}
\hfill
\end{minipage}
\vspace{-0.2in}
\caption{(a) Network $\mathcal{N}$ with $N = 10$. Number of interactions until absorption versus $L$ averaged over $5000$ different random initial states. (b) Network $\mathcal{N}$, with all nodes having a common value, except for two of them, one of which a unit higher and the other a unit lower than the common value. Number of interactions until absorption versus $N$ is depicted averaged over all possible initial positions and each with $250$  independent trials. The dashed lines represent the theoretical number of needed interactions computed in Theorem \ref{thm:timebound}.} 
\label{fig:time_Convergence}
\end{figure}

\section{Conclusion}\label{sec:conclusion}
In this paper we showed that by using the \DDD algorithm, desynchronization and , hence, TDMA scheduling is possible over discrete resources and in a decentralized manner. This approach can also be used to reach TDMA in more realistic situations, such as allocating resources based on each node's demand, or also considering the case where rise edge times of counters are not matched.
\appendices
\section{Proof of Lemma \ref{lemma:uniform_q}.}\label{App:uniform}
It was stated in Section \ref{sec:DDD} that by using uniform quantization on \eqref{eq:desync_psi1}, the algorithm may not converge. We continue our discussion on the \DDD equivalent consensus problem. Let $q_i = q_{i+1} + \Delta$, then by replacing the uniform quantization in \eqref{eq:active1_2} (note that $Q(\cdot)$ denotes the uniform quantization):
\begin{equation}
q_i^\prime =q_{i+1} + Q(\frac{1+\alpha}{2}\Delta) = q_{i+1}+\Delta.
\end{equation} 
where the right hand side equality is because for any $\alpha \in (0,1)$, $\frac{1}{2} < \frac{1+\alpha}{2} < 1$. Similarly, if $q_i = q_{i+1}-\Delta$, then 
\begin{equation}
q_i^\prime  = q_{i+1} + Q(-\frac{1+\alpha}{2}\Delta)= q_{i+1} - \Delta.
\end{equation}
In a network with $N = 4$, $\qb = (\Delta,\Delta+1,\Delta+2,\Delta+1)^T$ is a fixed point, but it is not a TDM state.
\section{Proof of Lemma \ref{lemma:DD_prop}}\label{app_1}
To prove statement 1) is true we observe that, since $q_i,q_{i+1} \in \mathcal{T}_L$, by replacing $q_{i+1} = q_i$ in equation (\ref{eq:active1}), $q_i^\prime = \Qcal{q_i} = q_i$, and consequently from (\ref{eq:active2}) $q_{i+1}^\prime = q_{i+1}$. 

To prove statement 2), we consider equation (\ref{eq:active1_2}). Two cases may happen. First, suppose $q_i - q_{i+1} = \Delta$. Then 
%
$q_i^\prime = q_{i+1} + \Qcalb{\frac{1+\alpha}{2}\Delta}$,
%
and since $0< \alpha < 1$, then $\frac{1}{2}<\frac{1+\alpha}{2} <1$. Thus:
\begin{equation}
q_i^\prime = 
\begin{cases}
q_{i+1} & \text{with probability } \frac{1-\alpha}{2}, \\
q_{i+1} + \Delta & \text{with probability } \frac{1+\alpha}{2}. \\
\end{cases}
\end{equation}
In either case, it can be inferred from (\ref{eq:active2}) that:
\begin{equation}
q_{i+1}^\prime = 
\begin{cases}
q_{i} & \text{if } q_{i}^\prime = q_{i+1}, \\
q_{i} - \Delta & \text{if } q_{i}^\prime = q_{i+1}+\Delta.
\end{cases}
\end{equation}
Thus in general, 
\begin{equation}
q_{i}^\prime - q_{i+1}^\prime =  
\begin{cases}
-\Delta & \text{with probability } \frac{1-\alpha}{2}, \\
\Delta & \text{with probability } \frac{1+\alpha}{2}.
\end{cases}
\end{equation}
Second, assume $q_{i+1} - q_{i} = \Delta$. 
Note that $\Qcal{-x} = -\Qcal{x}$, because of the fact that the distribution only depends on the distances and hence any transformation on $x$ which keeps distances unchanged (like taking the complement) has the same exact effect on $\Qcal{x}$. Using this fact, and equation (\ref{eq:active1_2}), the update can be written as 
%
$q_i^\prime = q_{i+1} - \Qcalb{\frac{1+\alpha}{2}\Delta}$.
%
With similar reasoning as before, the updated difference can be finally written as:
\begin{equation}
q_{i}^\prime - q_{i+1}^\prime =  
\begin{cases}
\Delta & \text{with probability } \frac{1-\alpha}{2}, \\
-\Delta & \text{with probability } \frac{1+\alpha}{2}.
\end{cases}
\end{equation}
Therefore this proves that, if $|q_i - q_{i+1}| = \Delta$, after the update $|q_i^\prime - q_{i+1}^\prime| = \Delta$.

For the proof of statement 3), suppose without loss of generality that $|q_i - q_{i+1}| = m\Delta$, $m \geq 2$ and $q_i > q_{i+1}$. By using equation (\ref{eq:active1}),
$q_i^\prime = \Qcalb{\frac{1+\alpha}{2}q_i + \frac{1-\alpha}{2}(q_{i} - m\Delta)}$,
and hence,
\begin{equation}\label{eq:lem_other1}
q_i^\prime = q_i -  \Qcalb{\frac{1-\alpha}{2}m\Delta}.
\end{equation}
Now, since $\frac{1-\alpha}{2}m\Delta$ is a positive value, the right hand side is greater than or equal to the left hand side, an hence $q_i^\prime \leq q_i$. The sum of the two values updated will be preserved after the update, thus, from equation (\ref{eq:active2}) it can be  readily concluded that $q_{i+1}^\prime \geq q_{i+1}$ and the first part of the statement is proved. For the second part, the proof is  the same as the first one because of the symmetry between the two cases.

Statement 4 combines the previous ones. If $q_i = q_{i+1}$ or $|q_i - q_{i+1}| = \Delta$. If $|q_i - q_{i+1}| = m\Delta$, $m \geq 2$,  suppose that $q_i > q_{i+1}$. Thus in equation (\ref{eq:lem_other1}), if $m^\prime = \lfloor \frac{1-\alpha}{2}m\rfloor$ and $p^\prime = \frac{(1-\alpha)m}{2} - m^\prime$ the second term in the right hand side can be written as
\begin{equation}\label{eq:lem_other2}
\Qcalb{\frac{1-\alpha}{2}m\Delta} = 
\begin{cases}
m^\prime\Delta & \text{with probability } 1 - p^\prime, \\
(m^\prime+1)\Delta & \text{with probability } p^\prime.
\end{cases}
\end{equation}
Thus, equation (\ref{eq:lem_other1}) can be manipulated as follows:
\begin{equation}\label{eq:lem_other3}
q_i^\prime - q_{i}= 
\begin{cases}
- m^\prime\Delta & \text{with probability }  1 - p^\prime, \\
-(m^\prime +1)\Delta & \text{with probability }  p^\prime.
\end{cases}
\end{equation}
A linear combination of equation (\ref{eq:active2}) and (\ref{eq:lem_other3}) is
\begin{equation}
q_i^\prime - q_{i+1}^\prime = 
\begin{cases}
(m - 2m^\prime)\Delta & \text{with probability }  1 - p^\prime, \\
(m -2m^\prime - 2)\Delta & \text{with probability }  p^\prime.
\end{cases}
\end{equation}
Since $0<\frac{1-\alpha}{2}<\frac{1}{2}$ and $0\leq m^\prime<\frac{m}{2}$,  $|q_i^\prime - q_{i+1}^\prime| \leq m\Delta = q_i - q_{i+1}$. Similarly,  if $q_i < q_{i+1}$ by the same approach it can be shown that  $|q_i^\prime - q_{i+1}^\prime| \leq m\Delta = q_{i+1}-q_i$, and this completes the proof.
\section{Proof of theorem \ref{thm:timebound}}\label{app_2}
\begin{figure*}[!t]
\normalsize
\setcounter{MYtempeqncnt}{\value{equation}}
\setcounter{equation}{40}
\begin{equation}
\bar{T} = \frac{N^4(\alpha+1)^2+N^3(\alpha+1)^2+12N^2(\alpha-1)^2 - 24N(\alpha-1)(2\alpha-1)+24(\alpha-1)^2}{24N(1-\alpha)(1+\alpha)}
\end{equation}
\setcounter{equation}{\value{MYtempeqncnt}}
\hrulefill
\vspace*{4pt}
\end{figure*}
For the proof of \ref{thm:timebound} we imitate the clever methodology used in \cite{quant} to calculate $\bar{T}(n_0)$. By definition $\bar{T}(n_0)$ is the expected number of interactions needed, in the worst case scenario, for network $\mathcal{G}$ to have its next compression, starting from iteration $n_0$. For a compression to happen, it is necessary that the active edge $(v_i,v_{i+1})$ is such that $|q_i - q_{i+1}| > \Delta$. Thus, the worst case (in expectation) is when there are the least possible number of potential nodes in the network $\mathcal{G}$ that can participate in a compression, i.e. only two nodes $v_i$ and $v_j$ have values $|q_i - q_j| > \Delta$ and all remaining values are equal. In this case, to fall in  the absorbing state, one would have to wait for a compression of the two outliers, while in all other interactions  the value of $V(\qb)$ will be unchanged. It is clear that any other initial state increases the probability of having a compression, and therefore this is the worst case scenario. It should be noted that if $\sum_{k=1}^{N}q_k(n_0) = L\Delta$, the stated worst case scenario happens only when $L = \ell N$ and all other cases (i.e. $L = \ell N +r$, $ 1 \leq r < N$) require less time in expectation.

Let us now focus on the worst case scenario, where there exists two nodes $v_i$ and $v_j$ such that:
\begin{equation}\label{eq:App_dist}
q_k(n_0) = 
\begin{cases}
(\ell+1)\Delta & k = i, \\
(\ell -1)\Delta & k = j, \\
\ell \Delta & k \neq i,j.
\end{cases}
\end{equation}
Any compression of the values that differ by $2\Delta$ from each other (initially at $v_i$ and $v_j$) will lead to a TDM state, which is an absorbing state. Otherwise, the outliers can change position, but there will always be only two of them. Thus, the evolution of the network can be viewed as a Markov chain with a single absorbing state ($\qb = \ell \Delta {\mathbf{1}}$) and $N(N-1)$ transient states that are all possible permutations of positions for the outliers around the ring graph $\mathcal{G}$. In the following we study this Markov chain, and derive how long it takes on average for these two values, which are outliers, to get into a compression thereby settling the network in a TDM state. 
\begin{figure}[htp]
\centering
\centerline{\includegraphics[width=1.0\linewidth]{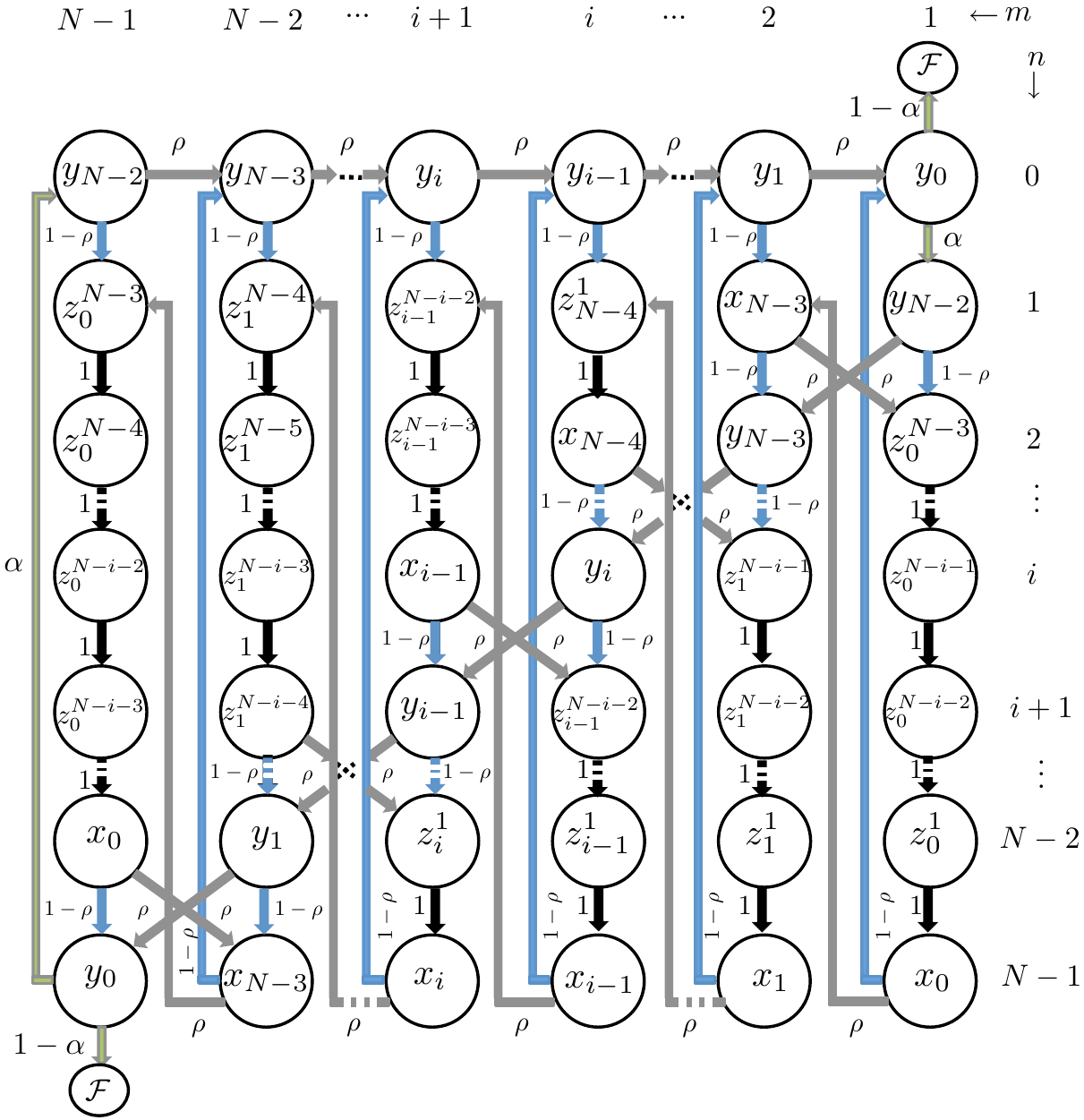}}
\hfill
\caption{The Markov chain for analysis of the absorption time of the proposed problem. The state space is $S = \mathcal{E}\cup \mathcal{F}$. Transition probabilities are also shown on each arc. The expected number of interactions starting from each state is written inside the corresponding circle.}
\label{fig:MC_app}
\end{figure}
The state space is defined as $S = \mathcal{E} \cup \mathcal{F}$, where $\mathcal{E} =\{(m,n): 1\leq m \leq N-1, 0 \leq n \leq N-1\}$, and $m$ represents the relative position of the node with value $(\ell+1)\Delta$ with respect to the node with value $(\ell-1)\Delta$ (e.g. $i-j \pmod N$ in the case of \eqref{eq:App_dist}), and $n$ represents  all the $N$ possible cases of the current active edge, where if $q_i = (\ell-1)\Delta$ then the active edge $(v_{i-1},v_{i})$ is assumed to be labeled as $0$ and others are labeled until $N-1$.  $\mathcal{F}$ denotes the absorbing state. So, without loss of generality we have fixed our reference to the node with value $(\ell-1)\Delta$ and the current state is identified based on that.
In Fig. \ref{fig:MC_app} this Markov chain is shown, and the circle in row $n$ and column $m$, represents the state $(m,n)$. In Fig. \ref{fig:MC_app}, the expected number of interactions until absorption starting from each state, is written in the corresponding circle. The transition probabilities of the Markov chain can be derived based on the update equations \eqref{eq:active1} and \eqref{eq:active2}. Recall that there is only one possible compression which can occur on the two outlier values, and it is possible only if both of their vertices are connected to the active edge. Based on the outlier values from (\ref{eq:App_dist}) and replacing them into equations (\ref{eq:active1}) and (\ref{eq:active2}), the probability of a compression in this case is $1-\alpha$ and, hence, the probability of a null interaction is $\alpha$. If the active edge has two nodes in consensus, the state moves forward with probability one to activating a new edge with a distance between the outliers unchanged. For all other cases where the interaction is between an outlier value and non-outlier one, we know from Corollary \ref{cor.prob} the probability of a swap is $\rho  \triangleq\frac{1-\alpha}{2}$ and the probability of a null interaction is $1-\rho$. There exists a symmetry in this Markov chain, which can be easily assessed by the choice of the shown variables $x$,$y$ and $z$ and the path through the absorbing states $\mathcal{F}$ in Fig. \ref{fig:MC_app}.
The expected number of interactions until absorption can be calculated as:
\begin{equation}\label{eq:exp_main}
\bar{T} = 2\Big[\sum_{i=0}^{N-3}x_i+\sum_{j=0}^{N-2}y_j + \sum_{k=0}^{N-4}\sum_{\ell = 1}^{N-3-k}z_k^\ell\Big]/(N(N-1))
\end{equation}
Based on the Markovian property, $z_k^\ell = 1+z_k^{\ell-1}$, and $z_k^1 = 1+x_k$ for all valid $k$ and $\ell$. By eliminating all the $z$ variables, the following equations can be written for all $x$ and $y$:
\begin{align}
& y_0 - \alpha y_{N-2} = 1, \label{eq:sol1} \\
& x_k - (1-\rho)y_k - \rho x_{N-k-3} = 1+k\rho,\label{eq:sol2} \\
& y_{k^\prime} - \rho y_{k^\prime-1} - (1-\rho)x_{N-k^\prime-2} = k^\prime - (k^\prime-1)\rho,\label{eq:sol3}
\end{align}
for all $0\leq k \leq N-3$ and $1\leq k^\prime \leq N-2$. Summing over all $k$ and $k^\prime$ in \eqref{eq:sol2} and \eqref{eq:sol3}, together with \eqref{eq:sol1}, $y_0$ and $y_{N-2}$ can be calculated as:
\begin{equation}
y_0 = \frac{\alpha(N-2)(N+1)+2}{2(1-\alpha)}, y_{N-2} = \frac{N(N-1)}{2(1-\alpha)}.
\end{equation}
The other unknowns can be calculated recursively for $1 \leq k \leq N/2$ as:
\begin{equation}\label{eq:rec1}
\begin{cases}
& x_{k-1} = \frac{D_k}{2(1-\rho)}+(y_{k-1}+y_{N-k-1})/2, \\
& x_{N-k-2} = \big(x_{k-1} - F_k - (1-\rho)y_{k-1}\big)/\rho,\\
& y_k = (k-(k-1)\rho)+\rho y_{k-1}+(1-\rho)x_{N-k-2}, \\
& y_{N-k-2} = \big(y_{N-k+1} - (1-\rho)x_{k-1} - G_k\big)/\rho,
\end{cases}
\end{equation}
where $D_k = (N-k-1)(2\rho-1)+F_k$, $F_k = 1+(k-1)\rho$, and $G_k = (N-k-1)-(N-k-2)\rho$. $x_0$ can be calculated based on \eqref{eq:rec1} and by using $y_0$ and $y_{N-2}$ as:
\begin{equation}
x_0 = \frac{N^2(\alpha+1)^2+N(3\alpha^2-6\alpha-1)-10\alpha^2+4\alpha+6}{4(1+\alpha)(1-\alpha)}.
\end{equation}
Considering all of the solutions by equations \eqref{eq:rec1}, it is not hard to see that $x_k -x_{k-1} = \frac{N^2(1+\alpha)-(2k+1)(1+\alpha)N -2\alpha+2}{4(1-\alpha)}$, and $y_k - y_{k-1} = \frac{N^2(\alpha+1)-N((2k+3)\alpha + (2k-1)) -2\alpha+2}{4(1-\alpha)}$. So, by using $x_0$ and $y_0$ we can calculate $\sum_{i=0}^{N-3}x_i$ and $\sum_{j=0}^{N-2}y_j$ as follows:
\begin{align}\label{eq:sum_x} 
\sum_{i=0}^{N-3}x_i &= (N-2)x_0 \\ &+ \frac{(N-2)(N-3)(N(N+1)(\alpha+1)+6(1-\alpha))}{24(1-\alpha)}, \nonumber 
\end{align}
\begin{align}\label{eq:sum_y}
\sum_{j=0}^{N-2}&y_j = (N-1)y_0 \\ &+\frac{(N-1)(N-2)(N^2(\alpha+1)-3N(3\alpha-1)+6(1-\alpha))}{24(1-\alpha)}. \nonumber 
\end{align}

Also, by doing some algebraic manipulations we have:
\begin{align}\label{eq:sum_z}
&\sum_{k=0}^{N-4}\sum_{\ell=1}^{N-3-k}z_k^\ell  = \sum_{k=0}^{N-4}\sum_{\ell=1}^{N-3-k}(\ell+x_k) \nonumber \\ &= \sum_{k=1}^{N-3}\Big[{\frac{k(k+1)}{2}+(N-2-k)x_{k-1}}\Big].
\end{align}
Substituting \eqref{eq:sum_x}, \eqref{eq:sum_y} and \eqref{eq:sum_z} into \eqref{eq:exp_main}, the expected number of interactions until absorption ($\bar{T}$) can be computed as in equation (41) at the top of the page. It can also be seen that for large N, $\bar{T} = \mathcal{O}\big(g(\alpha) N^3\big)$ where $g(\alpha) = \frac{\alpha+1}{24(1-\alpha)}$.
%
\bibliography{CPS-proposal,mac-references,cotdm}
\end{document}